\documentclass{amsart}[11pt]
\usepackage[margin=3.7cm]{geometry}
\usepackage{amsmath, amsfonts, amsthm, amssymb,mathtools}
\usepackage{subfigure}
\usepackage{caption} 
\usepackage{stmaryrd}
\usepackage{verbatim}
\usepackage{hyperref}
\usepackage{color}
\usepackage{float}
\usepackage{tikz}
\usetikzlibrary{matrix}
\usepackage{enumerate}
\numberwithin{equation}{section}
\mathtoolsset{showonlyrefs=true}
\newtheorem{theorem}{Theorem}[section]
\newtheorem{corollary}[theorem]{Corollary}

\newtheorem{proposition}[theorem]{Proposition}

\theoremstyle{definition}
\newtheorem{definition}[theorem]{Definition}
\newtheorem{remark}[theorem]{Remark}

\newtheorem{assumption}[theorem]{Assumption}

\newcommand{\ind}{1\hspace{-2.1mm}{1}}

\newcommand{\QQ}{\mathbb{Q}}
\newcommand{\RR}{\mathbb{R}}
\newcommand{\PP}{\mathbb{P}}
\newcommand{\TT}{\mathbb{T}}
\newcommand{\TTp}{\mathbb{T}}
\newcommand{\Vv}{\mathfrak{V}}

\newcommand{\NN}{\mathbb{N}}

\newcommand{\D}{\mathrm{d}}
\newcommand{\Bb}{\mathcal{B}}

\newcommand{\pd}{\partial}
\newcommand{\Dd}{\mathcal{D}}
\newcommand{\Ddg}{\mathcal{D}^{\gamma}}
\newcommand{\Ee}{\mathcal{E}}
\newcommand{\Gg}{\mathcal{G}}
\newcommand{\Ggm}{\mathcal{G}^{\Hm}}

\newcommand{\Corr}{\mathrm{Corr}}

\newcommand{\E}{\mathrm{e}}
\newcommand{\bbf}{\mathfrak{b}}
\newcommand{\Bf}{\mathfrak{B}}
\newcommand{\kf}{\mathtt{k}}
\newcommand{\Kf}{\mathtt{K}}
\newcommand{\ssf}{\chi}
\newcommand{\Gf}{\mathfrak{G}}
\newcommand{\Gfp}{\Gf^{+}}
\newcommand{\Gfm}{\Gf^{-}}

\newcommand{\EE}{\mathbb{E}}

\newcommand{\Ff}{\mathcal{F}}
\newcommand{\FF}{\mathbb{F}}

\newcommand{\Hm}{H_{-}}
\newcommand{\Hp}{H_{+}}

\newcommand{\rrho}{\overline{\rho}}
\newcommand{\Wf}{\boldsymbol{W}}

\newcommand{\Ddgi}{\reflectbox{$\Dd$}^\gamma}
\newcommand{\half}{\frac{1}{2}}

\graphicspath{{Figures/}}
\usepackage[textsize=tiny]{todonotes}

\definecolor{cadmiumgreen}{rgb}{0.0, 0.42, 0.24}

\begin{document}
\title{Risk premium and rough volatility}
\date{\today}
\author{Ofelia Bonesini}
\address{Department of Mathematics, London School of Economics and Political Science}
\email{o.bonesini@lse.ac.uk}

\author{Antoine Jacquier}
\address{Department of Mathematics, Imperial College London}
\email{a.jacquier@imperial.ac.uk}

\author{Aitor Muguruza}
\address{Department of Mathematics, Imperial College London and Kaiju Capital Management}
\email{aitor.muguruza-gonzalez15@imperial.ac.uk}

\maketitle
\tableofcontents

\begin{abstract}
One the one hand, rough volatility has been shown to provide a consistent framework to capture the properties of stock price dynamics both under the historical measure and for pricing purposes.
On the other hand, market price of volatility risk is a well-studied object in Financial Economics, and empirical estimates show it to be stochastic rather than deterministic.
Starting from a rough volatility model under the historical measure, we take up this challenge and provide an analysis of the impact of such a non-deterministic risk for pricing purposes.
\end{abstract}

%\tableofcontents

%%%%%%%%%%%%%%%%%%%%%%%%%%%%%%%%%%%%%%%
\section*{Introduction}

Rough volatility is a recent paradigm proposed by Gatheral, Jaisson and Rosenbaum~\cite{gatheral2022volatility}, 
which has attracted the attention of many academics and practitioners thanks to its numerous attractive properties~\cite{bayer2023rough}.
Despite some debate about whether volatility should be rough~\cite{jaber2022quintic,jaber2024volatility,guyon2023volatility,romer2022empirical},
this class of models provides a general framework to analyse 
both time series of the instantaneous volatility (under the historical measure~$\PP$)
and prices of financial derivatives (under the pricing measure~$\QQ$).
Starting from a rough version of the Bergomi model~\cite{Bergomi} under~$\PP$, 
Bayer, Friz and Gatheral~\cite{BFG15} showed that a deterministic market price of risk
preserved its structure under~$\QQ$ (somehow akin to the Heston model~\cite{heston1993closed} specification).

However, the Financial Economics literature has long shown that this market price of risk,
monitoring the transition from~$\PP$ to~$\QQ$ via Girsanov's transform,
is not constant nor deterministic but instead stochastic.
Its estimation has been the source of long academic discussions, 
outside the scope of the present paper though,
and we refer the interested reader to~\cite{bansal2002market, carr2016analyzing, chabi2012pricing, duan2014forward, kaido2009inference, mcdonald1984option} for some useful pointers.
This of course has serious practical implications for risk management,
and~\cite{graham2005long} are fascinating sources of information.
We focus here on this particular bridge between~$\PP$ and~$\QQ$ 
and show, not surprisingly, that the required stochasticity of the market price of risk unfortunately breaks the structure of the rough Bergomi model under~$\QQ$.
However, we link the H\"older regularity of the volatility process (lower in this class of rough models) with that of the change of measure, 
and design several specifications making the model tractable under~$\QQ$.
While the rough Bergomi model tracks well the behaviour of the historical volatility,
it is less powerful for option prices, especially when considering VIX smiles 
(which are more or less flat under this model).
Our new setup allows for more flexibility there, while preserving the $\PP$-tractability of the model.

Section~\ref{sec:Setup} provides the technical setup and analysis of the market price of risk, while the design of useful continuous-time rough stochastic volatility models with non-deterministic market prices of risk is detailed in Section~\ref{sec:Practical}.
Finally, in Section~\ref{sec:Extracting}, we perform an empirical analysis, 
estimating risk premia from historical options data.

%%%%%%%%%%%%%%%%%%%%%%%%%%%%%%%%%%%%%%%%%%%%%%
%%%%%%%%%%%%%%%%%%%%%%%%%%%%%%%%%%%%%%%%%%%%%%
\section{Rough volatility models and change of measure}
\label{sec:Setup}
Rough volatility models are a natural extension of classical stochastic volatility models. 
Starting from such a model under the historical measure~$\PP$, we characterise below its dynamics under martingale measures~$\QQ$ equivalent to~$\PP$, 
which then, by the fundamental theorem of asset pricing, allows for arbitrage-free option pricing.
Following~\cite{gatheral2022volatility} for example, we consider a rather general class of (rough) stochastic volatility models under~$\PP$, 
where the stock price process admits the following dynamics:
\begin{equation}\label{rough vol GJR}
    \left\{
    \begin{array}{rl}
        \displaystyle \frac{\D S_t}{S_t} & =  \displaystyle \mu_t \D t + \sqrt{v_t}\, \D W^\PP_t,\\
        v_t & = \displaystyle \psi(t,Y_t) , \\
        Y_t & = \displaystyle \int_{0}^t k(t,s) \D Z^\PP_s,
    \end{array}
    \right.
\end{equation}
starting from $S_0, v_0>0$, 
over a fixed time interval
$\TT := [0,T]$, for $T>0$.
Here, {$\mu$ is a $(\Ff_t)_{t\in\TT}$-measurable process, $\psi:\TT\times\RR\to (0,\infty)$ a continuous function, $\Wf^{\PP} = (W^{\PP}, W^{\PP, \perp})$ a two-dimensional standard Brownian motion  on a given filtered probability space $(\Omega,\Ff,\PP)$, with $\Ff := \Ff^{W^{\PP}}\vee \Ff^{W^{\PP, \perp}}$,
and $Z^{\PP} := \rho W^{\PP} + \rrho W^{\PP, \perp}$,
where $\rho \in [-1,1]$ and $\rrho:=\sqrt{1-\rho^2}$.
We assume that 
for each $t\in\TT$, the kernel $k(t,\cdot)$ is null on $\TT\setminus[0,t]$ and $\int_{0}^{\cdot} k(\cdot,s) \D Z^\PP_s$ is a well-defined continuous Gaussian process.
In particular, $k(t,\cdot) \in L^2([0,t])$ for each $t\in\TT$ ensures that the process is well defined, and we refer to~[Section~5]\cite{alos2001stochastic} or~\cite[Section~3]{decreusefond2005stochastic} for conditions of continuity criteria.
For example, the Gamma kernel, {commonly used to model turbulence, and  pioneered by  Barndorff-Nielsen and Schmiegel~\cite{Barndorff-Nielsen}}, is given by
\begin{equation}\label{eq:gammaKernel}
k(t,s)= (t-s)^{H-\half} \E^{-\beta(t-s)}\ind_{\{t\geq s\}}, 
\qquad \text{with }H\in(0,1), \quad \beta\geq0. 
\end{equation}
}
We further introduce the
set~$\FF_{b}$ of $\PP$-bounded and $(\Ff_t)_{t\in\TT}$-progressively measurable processes, {namely $X \in \FF_{b}$ if there exists a constant $c>0$ such that
$\PP\left(\sup_{t\geq 0} |X_t| \leq c\right)=1$},
and recall the Dol\'eans-Dade stochastic exponential of a square integrable {continuous} process~$X$:
$$
    \Ee(X)_t := \exp{\left( X_t -\half \langle X\rangle_t\right)}, \qquad \text{for }t \in \TT.
$$
{Finally, we introduce a progressively measurable interest rate process $(r_t)_{t \in \TT}$, define the corresponding Sharpe ratio by $\displaystyle \chi_t:= \frac{r_t-\mu_t}{\sqrt{v_t}}$, for $t\in \TT$ and consider the following assumption:

\begin{assumption}\label{assump_0}
The processes $\chi$ and $\gamma$ are $(\Ff_t)_{t\in\TT}$-adapted with c\`adl\`ag paths.
\end{assumption}
In order to state the main result, define the Radon-Nikodym derivative, for each $t\in \TTp$,
\begin{equation}\label{eq:RadonNik}
    \Ddg_t:=\frac{\D\QQ}{\D\PP}\bigg|_{\Ff_t^{\PP}}
    =\Ee\left(\int_{0}^{\cdot}\ssf_u\D W^\PP_u + \int_{0}^{\cdot}\gamma_u \D W^{\PP,{\perp}}_u\right)_t.
\end{equation}

\begin{proposition}\label{prop:local martingale}
Under Assumption~\ref{assump_0},
$\Ddg$ is a locally integrable $\PP$-local martingale on~$\TT$.
\end{proposition}

\begin{proof}
Since~$W^{\PP}$ and~$W^{\PP, \perp}$ are independent Brownian motions (and so locally square-integrable local martingales), thus by~\cite[Section~II, Theorem~20]{protter2005stochastic} the processes $X^1:=\int_0^{\cdot} \chi_u \D W^{\PP} _u$ and $X^2:=\int_0^{\cdot} \gamma_u \D W^{\PP, \perp} _u$ are locally square-integrable local martingales. 
The sum of two locally square-integrable local martingales is so itself and hence so is $X := X^1 + X^2$.
Finally, notice that the Radon-Nikodym derivative~~$\Ddg$ defined in~\eqref{eq:RadonNik} is precisely the stochastic exponential of $X$ and so a non-negative local martingale itself. This implies that~$\Ddg$ is a positive super-martingale which together with the fact that $\Ddg_0=1$ yields $\Ddg_t \in L^1$ for all $t \in \TT$.
\end{proof}
We finally consider the following set of assumptions,
in place for the rest of the paper:
\begin{assumption}\label{assu:General}\ 
\begin{enumerate}[(i)]
\item The function $\psi:\TT\times\RR\to (0,\infty)$ is continuous, bounded, and bounded away from the origin on $\TT\times(-\infty,a]$ for each $a>0$;
\item $K_\TT>0$ such that
$\sup_{t\in\TT} \left\{\int_{0}^t k(t,u)\lambda_u\D u\right\}\leq K_\TT$,
$\PP$-almost surely,
where~$\lambda$ denotes the market price of volatility risk defined by
\begin{equation}\label{eq:lambda}
    \lambda_t := \rho \ssf_t + \rrho\;\gamma_t;
\end{equation}
\item The correlation is negative: $\rho\leq 0$;
\item The Radon-Nikodym derivative satisfies $\EE[\Ddg_t]=1$ for all $t\in \TT$.
\end{enumerate}
\end{assumption}

\begin{remark}
    Note that Assumption~\ref{assu:General}(ii) also implies that, for all $n \in \NN$,
    \begin{equation}
        \sup_{t\in\TT} \left\{\int_{0}^t k(t,u)\lambda_u\D u\right\}\leq K_\TT,\qquad{\widehat{\PP}_n}\text{-almost surely},
    \end{equation}
    with
    \begin{equation}\label{eq: Pn Measure}
        \frac{\D\widehat{\PP}_n}{\D\PP}\bigg|_{\Ff_t}
        := \Ee\left(\int_{0}^{\cdot}\ssf_u\D W^\PP_u + \int_{0}^{\cdot}\gamma_u\D W^{\PP,{\perp}}_u\right)_{t\wedge\tau_n} \quad\text{and}\quad \tau_n := \inf\{t\geq 0,\; Y_t=n\}.
    \end{equation}
    Indeed, given an event~$A$ such that $\PP(A)=1$, Cauchy-Schwarz inequality yields
    \begin{align*}
        1
        & = \PP(A) = \EE[\bold{1}_A] = \widehat{\EE}_n[\bold{1}_A \Ee(X)^{-1}]
        \leq \widehat{\EE}_n[\bold{1}_A]^{\half} \widehat{\EE}_n[\Ee(X)^{-2}]^{\half}
        = \widehat{\PP}_n(A)^{\half} {\EE}[\Ee(X)^{-2}\Ee(X)]^{\half}\\
        & = \widehat{\PP}_n(A)^{\half} {\EE}[\Ee(X)^{-1}]^{\half}
        = \widehat{\PP}_n(A)^{\half} {\EE}\left[\Ee(-X)\exp\left(\int_0^{t \land \tau_n} \chi_u^2 + \gamma_u^2 \D u \right)\right]^{\half}
        \leq \widehat{\PP}_n(A)^{\half}, 
    \end{align*}
    where in the last step we have exploited the fact that~$\Ee(-X)$ is a non-negative supermartingale since~$-X$ a continuous local martingale.
    Thus, we have proved $\widehat{\PP}_n(A)=1$. 
    An analogous argument shows that the same holds for~$\QQ$.
\end{remark}

\begin{remark}\label{rem_eq_c}
    Assumption~\ref{assu:General}(iv) is guaranteed under different sets of stronger assumptions on $\gamma$, $\mu$, $\chi$, $r$ and~$v$, in particular
    \begin{itemize}
        \item if 
        $\displaystyle \EE\left[\exp\left\{\half \int_0^T |X_t|^2 \D t\right\}\right] <\infty$, namely~$X$ satisfy the Novikov condition;
        \item if Assumption~\ref{assump_0} and~\ref{assu:General}(i)-(ii) hold and  processes~$\mu, \gamma, r$ belong to~$\FF_{b}$, as detailed in Appendix~\ref{app:rem_eq_c}.
    \end{itemize}
\end{remark}
Proposition~\ref{prop:local martingale} justifies the use of a Dol\'eans-Dade exponential in the definition of~$\Ddg$.

\begin{theorem}\label{thm:changeOfMeasure}
Under Assumptions~\ref{assump_0} and~\ref{assu:General}, 
the following hold:
\begin{enumerate}[(I)]
\item the Radon-Nikodym derivative process~$\Ddg$ in~\eqref{eq:RadonNik}
is a true $\QQ$-martingale;
\item under the (arbitrage-free) equivalent risk-neutral martingale measure~$\QQ$, 
\begin{equation}\label{rough pricing Q}
\left\{
\begin{array}{rl}
\displaystyle \frac{\D S_t}{S_t} & =  \displaystyle r_t \D t + \sqrt{v_t} \,\D W^\QQ_t,\\
v_t & = \displaystyle \psi\left(t,\widehat{Y}_t +\int_{0}^t k(t,s)\lambda_s \D s\right),\\
\widehat{Y_t} & = \displaystyle \int_{0}^t k(t,s)\D Z^\QQ_s,
\end{array}
\right.
\end{equation}
with $S_0, v_0>0$, and~$\lambda$ is the market price of volatility risk in~\eqref{eq:lambda}
and where~$W^\QQ$ and~$Z^\QQ$ are $\QQ$-Brownian motions defined as
\begin{equation}\label{eq:QBM}
W^\QQ := \displaystyle W^\PP + \int_{0}^{\cdot}\ssf_u\D u
\qquad\text{and}\qquad
Z^\QQ := 
Z^\PP + \int_{0}^{\cdot} \lambda_u \D u;
\end{equation}
\item the discounted stock price $\widetilde{S} := \frac{S}{B}$ with $\D B_t = r_t B_t \D t $, $B_0=1$, is a true $\QQ$-martingale.
\end{enumerate}
\end{theorem}

\begin{proof}
To satisfy the no-arbitrage conditions, the change of measure for~$W^\PP$ is constrained by the martingale restriction on the discounted spot dynamics, 
while the Brownian motion~$Z^\PP$ gives freedom to the model and makes the market incomplete by the free choice of the process~$\gamma$. 
Consequently, the change of measure from~$\PP$ to~$\QQ$ and the corresponding Radon-Nikodym derivative directly follow from Girsanov's Theorem via~\eqref{eq:RadonNik}, provided that $\Ddg_{t}\in L^1$ and~$\Ddg$ is a true martingale. 
Thus, once we have shown~(I),
then~(II) automatically follows.
By Proposition~\ref{prop:local martingale}, $\Ddg_{t}\in L^1$ and, being a non-negative local martingale, 
is a super-martingale, and a true martingale on~$\TTp$ if and only if $\EE\left[\Ddg_{T}\right]=1$. This is guaranteed by~Assumption~\ref{assu:General}(iv), hence~(I) holds, and therefore~(II) as well.

\noindent We now prove (III): \textit{the discounted price $\widetilde{S} = \frac{S}{B}$ is a true martingale for $\rho\leq 0$}.
It\^{o}'s formula under~$\QQ$ yields
\begin{equation}\label{eq:disc_price}
\left\{
\begin{array}{rl}
\displaystyle \frac{\D \widetilde S_t}{ \widetilde S_t}
& =  \sqrt{v_t} \,\D W^\QQ_t,\\
v_t & = \displaystyle \psi\left(t,\widehat{Y}_t +\int_{0}^t k(t,s)\lambda_s \D s\right),\\
\widehat{Y_t} & = \displaystyle \int_{0}^{t} k(t,s)\D Z^\QQ_s.
\end{array}
\right.
\end{equation}
Define the stopping time $\iota_n:=\inf\{t\geq 0,\; \widehat{Y}_t = n\}$. 
For any $t\in\TT$, the random function $g(x):=\psi\left(t, x+\int_{0}^t k(t,s)\lambda_s \D s\right)$ 
is bounded $\QQ$-almost surely on $(-\infty, a]$ by Assumption~\ref{assu:General}(i)-(ii), 
with~$\lambda$ in~\eqref{eq:lambda}, so that, since $\widetilde S$ is a $\QQ$-local martingale:
\begin{align}
\widetilde S_0
 = \EE^{\QQ}[\widetilde S_{T\wedge \iota_n}]
 = \EE^{\QQ}[\widetilde S_T \ind_{\{T< \iota_n\}}]+\EE^{\QQ}[\widetilde S_{\iota_n} \ind_{\{T> \iota_n\}}].
\end{align}
The first term converges to~$\EE^{\QQ}[\widetilde S_{T}]$ as~$n$ tends to infinity, hence
\begin{align}
\widetilde S_0-\EE^{\QQ}[\widetilde S_{T}]=\lim_{n\uparrow\infty}\EE^{\QQ}[\widetilde S_{\iota_n} \ind_{\{T> \iota_n\}}].
\end{align}
Girsanov's Theorem further gives
$\EE^{\QQ}[\widetilde S_T \ind_{\{T> \iota_n\}}]= \widetilde S_0\widetilde{\PP}_n(T> \iota_n)$, where $\widetilde{\PP}_n$ is such that
$\widetilde{W}^n_t=W^\QQ_t-\int_{0}^{t\wedge \iota_n}v_s \D s
$
is a $\widetilde{\PP}_n$-Brownian motion. 
Note that, for $t<\iota_n$, $\widehat{Y}_t = \widetilde{Y}_t + \rho \int_{0}^t k(t,s)v_s  \D s$, where $\widetilde{Y}_t=\int_{0}^t k(t,s) \D \widetilde{Z}^n_s$ and $\widetilde{Z}^n_t:=Z^\QQ_t-\rho \int_{0}^t k(t,s)v_s  \D s$ is also a $\widetilde{\PP}_n$-Brownian motion. 
We conclude that, if $\rho\leq 0$, then $\widehat{Y}_t\geq \widetilde{Y}_t$ and
$$
\lim_{n\uparrow\infty}\widetilde{\PP}_n(\iota_{n}\leq T)
 \leq \lim_{n\uparrow\infty}\widetilde{\PP}_n(\widetilde{\iota}_n\leq T)
 = \lim_{n\uparrow\infty}\PP\left(\sup_{t\in[0,T]} Y_t\leq n\right)
 =0,
$$
where $\widetilde{\iota}_n:=\inf\{t\geq 0,\; \widetilde{Y}_t=n\}$, and hence~$\widetilde S$ is a true martingale.
\end{proof}

}

\begin{remark}
Under Assumption~\ref{assu:General}, 
consider $\rho\leq 0$, some valid function~$\psi$ and kernel~$k$,
and the constant values $\gamma_s=\overline{\gamma}$ and $\overline{\mu}=\mu_s\leq r_s=\overline{r}$ for $\overline{\gamma},\overline{\mu},\overline{r}\in\mathbb{R}$
ensuring Assumption~\ref{assu:General}(iv), 
so that Theorem~\ref{thm:changeOfMeasure} applies.
In this scenario, a sufficient condition for the change of measure to be well defined is that the physical drift must be smaller than the risk-free rate. 
\end{remark}

%%%%%%%%%%%%%%%%%%%%%%%%%%%%%%%%%%%%%%%%%%%%%%%%%%%%%
%%%%%%%%%%%%%%%%%%%%%%%%%%%%%%%%%%%%%%%%%%%%%%%%%%%%
\subsection{Rough volatility models via Generalised Fractional Operators}
{Many rough volatility models can be represented~\cite{horvath2024functional} in terms of Generalised Fractional Operators (GFOs), which are defined as follows~\cite[Definition 1.1]{horvath2024functional}:}
\begin{definition}\label{def:GFO}
For any $\beta\in(0,1)$, $\alpha\in(-\beta,1-\beta)$ and $h\in\mathcal{C}^1_b((0,\infty))$ such that $h'(\cdot)\leq 0$,
the GFO associated to the kernel $k(x):=x^{\alpha}h(x)$ applied to $f\in C^\beta(\RR)$ is defined as
\begin{equation}\label{eq:GFO}
(\Gg^\alpha f)(t) := 
\left\{
\begin{array}{ll}
\displaystyle \int_{0}^{t}(f(s)-f(0))\frac{\D}{\D t}k(t-s)\D s, & \text{if } \alpha \in [0,1-\beta),\\
\displaystyle \frac{\D}{\D t}\int_{0}^{t}(f(s)-f(0))k(t-s)\D s, & \text{if } \alpha \in (-\beta, 0).
\end{array}
\right.
\end{equation}
\end{definition}

To simplify future notations, 
we let $H_{\pm} := H \pm \half$
for $H \in (0, \half)$.
We now introduce a specific setup that will drive the rest of our computations:
consider the power-law kernel
\begin{equation}\label{eq:kernelPower}
\kf_{\alpha}(u) := u^{\alpha}\ind_{\{u\geq 0\}},
\end{equation}
as well as the set
$\displaystyle \Lambda_{\beta,H} :=
\Big\{\lambda\in\mathcal{C}^\beta
\text{ for some $\beta\in(0,1]$ such that $H_- \in (-\beta, 0)$ and $\lambda(0) = 0$}
\Big\}$.
To this particular power-law kernel, the GFO 
(from Definition~\ref{def:GFO}, since $\Hm \in (-\half,0)$) reads
$$
(\Ggm f)(t) = \frac{\D}{\D t}
\int_{0}^{t}(f(s)-f(0))\kf_{\Hm}(t-s)\D s.
$$
Denote further
$$
\Kf(t,s):=\int_{0}^{t} \kf_{\Hm}(u-s)\D u
= \frac{\kf_{\Hp}(t-s)}{H_+},
$$
so that the corresponding GFO is precisely $\frac{1}{\Hp}\Gg^{\Hp}$.
To streamline notations and emphasise nice symmetries, we introduce the notations
\begin{equation}\label{eq:niceGFO}
\Gf^{-} := \Gg^{\Hm}
\qquad\text{and}\qquad
\Gf^{+} := \frac{1}{\Hp}\Gg^{\Hp}.
\end{equation}

From the properties of GFO~\cite[Proposition~1.2]{horvath2024functional},
then $\Gfp\lambda\in\mathcal{C}^{\beta+H_+}$ 
as soon as $\lambda \in \Lambda_{\beta,H}$.

\begin{corollary}[GFO representation]\label{cor:changeOfMeasure}
With the kernel~$\kf_{\Hm}$ in~\eqref{eq:kernelPower} and $\lambda \in \Lambda_{\beta,H}$, 
the system~\eqref{rough pricing Q} under the risk-neutral measure~$\QQ$ can be rewritten as
\begin{equation*}
\left\{\begin{array}{rl}
\displaystyle \frac{\D S_t}{S_t} & = \displaystyle r_t \D t + \sqrt{v_t} \D W^\QQ_t,\\
v_t & = \displaystyle  \psi\Big(t,(\Gfm Z^\QQ)(t) + (\Gfp\lambda)(t)\Big).
\end{array}\right.
\end{equation*}
\end{corollary}

\begin{proof}
The fact that 
$\displaystyle
\int_{0}^\cdot \kf_{\Hm}(\cdot-s)\D Z^\QQ_s = \Gfm Z^\QQ$
is straightforward by the properties of GFO in~\cite[Proposition 1.4]{horvath2024functional}. 
Furthermore, for any $\lambda \in \Lambda_{\beta,H}$ and any $t\in\TT$,
$$
\int_0^t \kf_{\Hm}(t-s) \lambda_s \D s
= \int_0^t \frac{\D}{\D t}\Kf(t,s) (\lambda_s-\lambda_0) \D s
= \left(\Gfp\lambda\right)(t).
$$
\end{proof}
Note that since $\Gfp\lambda\in\mathcal{C}^{\beta+H_+}$, then the risk premium has sample paths with H\"older regularity greater than $\half$, regardless of the value of~$H$.

%%%%%%%%%%%%%%%%%%%%%%%%%%%%%%%%%%%%%%%%%%%%%%
%%%%%%%%%%%%%%%%%%%%%%%%%%%%%%%%%%%%%%%%%%%%%
\section{Modelling the risk premium process: A practical approach}
\label{sec:Practical}

In practice, the process~$\lambda$ is directly modelled without resorting to a change of measure starting from~$\gamma$. 
We now consider different modelling choices for the risk premium~$\lambda$ and analyse some of its properties.
In spite of the formal derivation of Theorem~\ref{thm:changeOfMeasure}, a numerical treatment of the integral $\int_{0}^t \bullet \; \D s $ is rather intricate. To overcome this issue, Bayer, Friz and Gatheral~\cite{BFG15} elegantly came up with the forward variance form of rough volatility in the spirit of Bergomi~\cite{Bergomi}.
We shall restrict ourselves to this functional form {(defined below in~\eqref{rough pricing Q lognormal})} for the remainder of the section.
Consider~\eqref{rough pricing Q} with $\psi(t,x)=\xi_0(t)\E^{\nu x}$, 
with $\xi_0(t):=\EE[v_t|\mathcal{F}_0]$ and $\nu>0$.
Then the risk-neutral dynamics in forward variance form read
\begin{equation}\label{rough pricing Q lognormal}
\left\{
\begin{array}{rl}
\displaystyle \frac{\D S_t}{S_t} & =  \displaystyle r_t \D t + \sqrt{v_t} \,\D W^\QQ_t,\\
v_t & = \displaystyle \xi_0(t)\exp\left(
\nu\left(\int_{0}^{t} \kf_{\Hm}(t-s)\D Z^\QQ_s
 + \int_{0}^{t} \kf_{\Hm}(t-s) \lambda_s \D s\right)\right ).
\end{array}
\right.
\end{equation}
In the remaining of this section, 
the process~$X^{\QQ}$ will denote a $\QQ$-Brownian motion possibly correlated 
with~$W^{\QQ}$ and~$Z^{\QQ}$.

\subsection{Risk premium driven by It\^o diffusion} 

Generalised Fractional Operators provide a natural framework to model risk premium processes driven by diffusions. 
The statement below shows the details of such a construction.
Recall that the Beta function is defined as 
$\Bf(x,y) := \int_{0}^{1}s^{x-1}(1-s)^{y-1}\D s$, for $x,y>0$.

\begin{proposition}\label{SDE lambda}%[General stochastic lambdas]
For $H \in (0,\half)$, 
$\alpha\in\left(-\half,0\right)$, 
let
$\lambda := \bbf\Gg^{\alpha}Y^{\QQ} \in\mathcal{C}^{\alpha+\half}$
with $\bbf := \Bf(H_+, \alpha +1)^{-1}$
and 
$
Y^{\QQ}_t = \int_{0}^{t}b(s,Y^{\QQ}_s)\D s + \int_{0}^{t}\sigma(s,Y^{\QQ}_s)\D X^\QQ_s$,
where $b(\cdot)$ and $\sigma(\cdot,\cdot)$ satisfy the Yamada-Watanabe conditions~\cite[Section~5.2, Proposition 2.13]{karatzas2012brownian} 
for pathwise uniqueness ensuring a weak solution. 
Then $\Gg^{\alpha+ H_+}Y^{\QQ}\in\mathcal{C}^{H+\alpha+1}$ and 
\begin{equation}\label{eq: operator representation}
v_t=\xi_0(t)\exp\Big\{ \nu \left((\Gfm Z^{\QQ})(t)
+(\Gg^{\alpha+ H_+}Y^{\QQ})(t)\right) \Big\}.
\end{equation}
Furthermore, if $Y^{\QQ} = X^\QQ$ and $\D\langle Y^{\QQ}, Z^\QQ\rangle_t = \rho\,\D t$ with $\rho\leq 0$, then
\begin{align}\label{eq:op_rep_cond_V}
\EE^{\QQ}\left[v_t|\mathcal{F}_s\right]
&=\xi_0(t)\exp\Big\{ \nu \left[\left(\Gfm Z^{\QQ}\right)(s,t)
+\left(\Gg^{\alpha+ H_+}X^{\QQ}\right)(s,t)\right]\Big\} \\
&\times\exp\left\lbrace\frac{\nu^2}{2}\left(\frac{\kf_{2H}(t-s)}{2H}+\frac{\kf_{2(H+1)}(t-s)}{2H_+^2 (H+1)} + \rho \frac{\kf_{2\Hp}(t-s)}{H_+^2} \right) \right\rbrace,
\end{align}
where
$\left(\Gg^{\Hm} Z^{\QQ}\right)(s,t) := \int_{0}^{s} \kf_{\Hm}(t-u) \D Z^\QQ_u$ for $0\leq s\leq t$,
and similarly for $\left(\Gg^{\alpha+ H_+}X^{\QQ}\right)(s,t)$.
\end{proposition}

\begin{proof}
We first prove~\eqref{eq: operator representation},
which follows from~\cite[Proposition 1.2]{horvath2024functional} and the identities highlighted above in Corollary~\ref{cor:changeOfMeasure}. 
Indeed, in view of Corollary~\ref{cor:changeOfMeasure} we only need to show that
$\int_0^t \kf_{\Hm}(t-s) \lambda_s \D s 
= (\Gg^{\alpha+ H_+}Y^{\QQ})(t)$.
Replacing the expression for~$\lambda$ in the integral and using stochastic Fubini theorem, we obtain
\begin{align*}
\int_0^t \kf_{\Hm}(t-s) \lambda_s \D s 
 = \bbf \int_0^t \kf_{\Hm}(t-s) (\Gg^{\alpha}Y^{\QQ})(s)  \D s
 & = \bbf \int_0^t \kf_{\Hm}(t-s) 
 \int_0^s \kf_{\alpha}(s-u) \D Y^{\QQ}_u  \,\D s \\
 & = \bbf \int_0^t \int_u^t \kf_{\Hm}(t-s) \kf_{\alpha}(s-u) \D s \, \D Y^{\QQ}_u.
\end{align*}
Now, direct computations for the inner integral yield
$$
\int_{u}^{t}\kf_{\Hm}(t-s) \kf_{\alpha}(s-u)  \D s
 = \kf_{\alpha+\Hp}(t-u)
 \int_{0}^{1} (1-s)^{H_-} s^\alpha  \D s
 = \Bf(\alpha+1, \Hp) \kf_{\alpha+\Hp}(t-u).
$$
Therefore
$$
\int_0^t \kf_{\Hm}(t-s) \lambda_s \D s 
 = \bbf \int_0^t \Bf(\alpha+1, \Hp) \kf_{\alpha+\Hp}(t-u) \D Y^{\QQ}_u
 = \int_0^t \kf_{\alpha+\Hp}(t-u) \D Y^{\QQ}_u
 =(\Gg^{\alpha+ H_+}Y^{\QQ})(t). 
$$
We now move to the proof of the identity~\eqref{eq:op_rep_cond_V}.
Exploiting the representation of~$v_t$ in this specific case and the measurability and independence properties of the Brownian increments,
\begin{align*}
\EE^{\QQ}\left[v_t|\mathcal{F}_s\right]
& = \xi_0(t)\EE^{\QQ}_s\left[\exp\left\lbrace \nu \left[\left(\Gg^{H_-}Z^{\QQ}\right)(t)
+\left(\Gg^{\alpha+ H_+}X^{\QQ}\right)(t)\right] \right\rbrace\right]\\
& = \xi_0(t)\EE^{\QQ}_s\left[\exp\left\lbrace \nu \left[\int_0^t \kf_{\Hm}(t-u) \D Z^\QQ_u
+\int_0^t \kf_{\alpha+\Hp}(t-u)\D X^{\QQ}_u\right] \right\rbrace\right]\\
& = \xi_0(t)\exp\left\lbrace \nu \left[\int_0^s \kf_{\Hm}(t-u) \D Z^\QQ_u
+\int_0^s \kf_{\alpha+\Hp}(t-u)\D X^{\QQ}_u\right] \right\rbrace\\
& \qquad \times \EE^{\QQ}_s\left[\exp\left\lbrace \nu \left[\int_s^t \kf_{\Hm}(t-u) \D Z^\QQ_u
+\int_s^t \kf_{\alpha+\Hp}(t-u)\D X^{\QQ}_u\right] \right\rbrace\right]\\
& = \xi_0(t)\exp\left\lbrace \nu \left[\int_0^s \kf_{\Hm}(t-u) \D Z^\QQ_u
+\int_0^s \kf_{\alpha+\Hp}(t-u)\D X^{\QQ}_u\right] \right\rbrace\\
& \qquad \times\exp\left\lbrace\frac{\nu^2}{2}\left(\int_s^t \kf_{2\Hm}(t-u) \D u +\int_s^t \frac{\kf_{2H+1}(t-u)}{H_+^2} \D u+\rho\int_s^t \frac{\kf_{2H}(t-u)}{H_+} \D u \right) \right\rbrace\\
& = \xi_0(t)\exp\left\lbrace \nu \left[\int_0^s \kf_{\Hm}(t-u) \D Z^\QQ_u
+\int_0^s \kf_{\alpha+\Hp}(t-u)\D X^{\QQ}_u\right] \right\rbrace\\
& \qquad \times\exp\left\lbrace\frac{\nu^2}{2}\left(\frac{\kf_{2H}(t-s)}{2H}+\frac{\kf_{2(H+1)}(t-s)}{2H_+^2 (H+1)} + \rho \frac{\kf_{2\Hp}(t-s)}{H_+^2} \right) \right\rbrace.
\end{align*}
Thus we only have to show that 
$$
\left(\Gg^{H_-}Z^{\QQ}\right)(s,t) = \int_0^s \kf_{\Hm}(t-u) \D Z^\QQ_u
\qquad\text{and}\qquad
\left(\Gg^{\alpha+ H_+}X^{\QQ}\right)(s,t) = \int_0^s \kf_{\alpha+\Hp}(t-u)\D X^{\QQ}_u.
$$
We prove the first identity, the second being analogous. 
It is a straightforward consequence of the definitions and the properties of Brownian increments:
\begin{align*}
\left(\Gg^{H_-}Z^{\QQ}\right)(s,t) 
 := \EE^\QQ_s\left[
 \int_{0}^{t} \kf_{\Hm}(t-u) \D Z^\QQ_u\right]
 & = \int_0^s \kf_{\Hm}(t-u)\D Z^\QQ_u + \EE^\QQ_s \left[ 
 \int_{s}^{t} \kf_{\Hm}(t-u) \D Z^\QQ_u\right]\\
 & = \int_0^s \kf_{\Hm}(t-u) \D Z^\QQ_u.
\end{align*}
\end{proof}

\begin{remark}
Since the instantaneous variance in this model is log-Normal, the results in~\cite[Proposition~3.1]{jacquier2018vix} and numerical methods therein still apply for the VIX with minimal changes.
\end{remark}

%%%%%%%%%%%%%%%%%%%%%%%%%%%%%%%%%%%%%%%%%%
\subsection{A risk premium driven by a CIR process}

A second natural choice is to consider the Cox-Ingersoll-Ross (CIR) process
\begin{align}\label{eq:CIR_YQ}
    \D Y^{\QQ}_s = \kappa(\theta -Y^{\QQ}_s) \D s+\sigma \sqrt{Y^{\QQ}_s}\ \D X^{\QQ}_s,
\end{align}
with $\kappa, \theta, \sigma>0$.
As tempting as this approach might seem, it is not trivial at all to compute the basic quantity $\EE^\QQ[v_t]$ here, 
as the following proposition shows.

\begin{proposition}\label{prop: MGF fCir}
Assume that the Brownian motions $Z^{\QQ}$ and $X^{\QQ}$ are independent and consider 
$\lambda = \Gg^{\alpha}Y^{\QQ} \in\mathcal{C}^{\alpha+\half}$, with $Y^\QQ$ defined in~\eqref{eq:CIR_YQ}. 
Then, for any $s\leq t$,
\begin{align*}
 \EE_{s}^{\QQ}[v_t]
 = \xi_0(t)& \exp\Big\{ \nu \left[\left(\Gfm Z^{\QQ}\right)(s,t)
+\left(\Gg^{\alpha+ H_+}Y^{\QQ}\right)(s,t)\right] \Big\}\\
 & \exp\left\{\frac{\nu^2}{2}\int_{s}^{t} \kf_{\Hm}(t-u)^2\D u
-Y^{\QQ}_s C(s,T) - A(s,T)\right\},
\end{align*}
where $A(t,T) := -\kappa\theta \int_t^TC(u,T)\D u$ and~$C$ satisfies the Riccati equation
$$
\nu \kf_{\Hm}(T,t) - \pd_{t}C(t,T) + C(t,T)\theta+\frac{\sigma^2}{2}C^2(t,T) = 0, 
$$
for $t \in [0,T)$, 
with boundary condition $C(T,T)=0$.
\end{proposition}

\begin{proof}
By independence of the driving Brownian motions we have, 
for any $u\leq t$,
\begin{align*}    \EE[v_t|\Ff_u]
 & = \xi_0(t)\exp\left\lbrace \nu \left(\left(\Gfm Z^{\QQ}\right)(u,t)
+\left(\Gg^{\alpha+\Hp}Y^{\QQ}\right)(u,t)\right) \right\rbrace\\
 & \qquad \times \EE\left[
 \exp\Big\{ \nu \left(\left(\Gfm Z^{\QQ}\right)(t)-\left(\Gfm Z^{\QQ}\right)(u,t)\right)\Big\}\right]\\
&\qquad \times \EE\left[\exp\Big\{\nu\left(\left(\Gg^{\alpha+\Hp}Y^{\QQ}\right)(t)-\left(\Gg^{\alpha+\Hp}Y^{\QQ}\right)(u,t)\right)\Big\} \right],
\end{align*}
where the first expected value is the MGF of a Gaussian random variable, hence
$$
\EE\left[\exp\Big\{ \nu \Big(\left(\Gfm Z^{\QQ}\right)(t)-\left(\Gfm Z^{\QQ}\right)(u,t)\Big)\Big\}\right]
 = \exp\left\{\frac{\nu^2}{2}\int_u^t \kf_{\Hm}(t,s)\D s\right\}.
$$
We are now interested in computing the second expectation
$$
\EE\left[\exp\Big\{\nu\left(\left(\Gg^{\alpha+\Hp}Y^{\QQ}\right)(t)- \left(\Gg^{\alpha+\Hp}Y^{\QQ}\right)(u,t)\right)\Big\}\right],
$$
where
$\displaystyle
\left(\Gg^{\alpha+\Hp}Y^{\QQ}\right)(t) - \left(\Gg^{\alpha+\Hp}Y^{\QQ}\right)(u,t)
 = \int_u^t \kf_{\Hm}(t,s)Y^{\QQ}_s \D s$.
This is, in spirit, similar to computing a bond price in the CIR model. To do so, define 
\begin{equation}\label{eq:MGF_B}
B(t,T)        
 := \EE\left[\exp\left(\nu\int_t^T \kf_{\Hm}(T,s) Y^{\QQ}_s\D s\right)\bigg|\Ff_t\right],
\end{equation}
where $t\leq T$. We note that $B(\cdot,T)$ is a semimartingale as $T$ is fixed, therefore applying the conditional version of Feynman-Kac's formula, we obtain 
\begin{equation}\label{eq:CIRPDE}
\left(\nu r \kf_{\Hm}(T,t) + \pd_t+ \kappa(\theta - y)\pd_r + \frac{\sigma^2}{2} r\pd_{yy}\right)\Bb (y,t,T) = 0.
\end{equation}
where $\Bb(y,t,T)$ is such that $\Bb(Y^{\QQ}_t,t,T) = B(t,T)$
given in~\eqref{eq:MGF_B}.
With an ansatz of the type $\Bb(y,t,T) = \exp\{-y C(t,T)-A(t,T)\}$, we have at $(y,t,T)$,
\begin{align*}
    & \pd_{t}\Bb(y,t,T) = -(y \pd_{t} C(t,T) + \pd_{t}A(t,T))\Bb(y,t,T),\\
    &\pd_{y}\Bb(y,t,T) = -C(t,T)\Bb(y,t,T),\qquad
\pd_{yy}\Bb(y,t,T) = C(t,T)^2 \Bb(y,t,T),
\end{align*}
and the PDE~\eqref{eq:CIRPDE} becomes, with $r=Y^{\QQ}_t$,
$$
\left(\nu Y^{\QQ}_t \kf_{\Hm}(T,t) - \left(Y^{\QQ} \pd_{t}C + \pd_{t}A + \kappa\left(\theta - Y^{\QQ}_t\right)C\right)
+ \frac{\sigma^2}{2}C^2 Y^{\QQ}_t \right) \Bb(Y^{\QQ}_t,t,T) = 0,
$$
which further simplifies to
$$
\left(\nu \kf_{\Hm}(T,t) - \pd_{t}C - \kappa C + \frac{\sigma^2}{2}C^2\right)Y^{\QQ}_t\Bb(Y^{\QQ}_t,t,T) - (\kappa\theta C + \pd_{t}A)\Bb(Y^{\QQ}_t,t,T) = 0.
$$
The last term cancels for
$A(t,T) = -\kappa\theta \int_t^TC(u,T)\D t$,
and the Riccati equation
$\nu \kf_{\Hm}(T,t) - \pd_{t}C(t,T) -\kappa C(t,T) + \frac{\sigma^2}{2} C^2(t,T) = 0$
remains, 
with $A(T,T)=C(T,T)=0$.
\end{proof}

Already in the uncorrelated case the computation of $\EE^\QQ[v_t]$ becomes very costly, having to solve a quadratic ODE (with time-dependent coefficients) for each time~$t$. 
In the correlated case there is no hope to obtain any semi-analytic result since one would need to compute cross terms and there is no tool coming from It\^o's calculus available in that case.
{
The approach in this section was essentially \textit{top-down}, 
meaning that we specified a form for the market price of volatility risk~$\lambda$ and deduced the shape of the model with this specification.
Unfortunately, our analysis showed that this may ultimately not be so successful as the final form of the model is rather intricate, probably too much so for practical purposes.
Alternatively, one may first infer some shape of~$\lambda$ from market data 
(short rate of interest, expected
returns and instantaneous volatility)
and then use it to price options under~$\QQ$.
}

%%%%%%%%%%%%%%%%%%%%%%%%%%%%%
%%%%%%%%%%%%%%%%%%%%%%%%%%%%%

\section{Roughly extracting the Risk Premium from the Market}
\label{sec:Extracting}

We now consider the risk premium process~$\lambda$  deterministic, 
and obtain a formula linking~$\PP$ and~$\QQ$ market observable quantities. 
The following theorem shows how to infer the risk premium from the market using forecasts 
under~$\PP$ and Variance Swap prices under~$\QQ$.

\begin{theorem}\label{thm: risk premium}
    Consider the rough volatility model~\eqref{rough vol GJR} under~$\PP$. 
    If $\psi(t,x)=\xi_0(t)\E^{\nu x}$, 
    $\mu_s=r_s$ for all $s\geq 0$ and 
    $\left(\lambda_s\right)_{s\geq 0}\in L^2(\RR)$ is deterministic, then
    \begin{equation}\label{eq:riskPremium}
        \nu\rrho\int_s^t \kf_{\Hm}(t,u)\gamma_u \D u=\log\left(\frac{\EE^\QQ[v_t|\Ff_s]}{\EE^\PP[v_t|\Ff_s]}\right)=\log\left(\frac{\xi_s(t)}{\EE^\PP[v_t|\Ff_s]}\right).
    \end{equation}
\end{theorem}

\begin{proof}

If $\mu=r$ almost surely, the Radon-Nikodym derivative~\eqref{eq:RadonNik} in Theorem~\ref{thm:changeOfMeasure} reads
$\Ddg = \Ee\left(\int_{0}^{\cdot}\gamma_s \D W^{\PP\perp}_s \right)$,
with $\lambda_s = \rrho\gamma_s$, and the inverse Radon-Nikodym derivative is given by
$\Ddgi := \frac{1}{\Ddg} = \Ee\left(-\int_{0}^{\cdot}\gamma_s \D W^{\QQ\perp}_s \right)$.
Then, the conditional change of measure formula yields	
    \begin{equation}\label{eq: Conditional Girsanov}
        \EE^\PP[v_t|\Ff_s]
        =\frac{\EE^\QQ\left[v_t\Ddgi_t |\Ff_s\right]}{\EE^\QQ\left[\Ddgi_t |\Ff_s\right]}.
    \end{equation}
    On the one hand,    $\EE^\QQ\left[\Ddgi_t|\Ff_s\right]=\Ee\left(-\int_{0}^{\cdot}\gamma_u \D W^{\QQ\perp}_u\right)_s$
    by the properties of the  stochastic exponential and Gaussian moment generating functions. 
    On the other hand, since, for $t\in \TT$, $Z^{\QQ}_t = Z^{\PP}_t + \int_{0}^t \lambda_s \D s$ and~$\lambda$ is deterministic,  then
\begin{align}\label{eq:vtp-q}
        \EE^\QQ\left[v_t\Ddgi_t |\Ff_s\right] 
        & =\EE^\QQ\left[\left.\exp\left\{\nu\left(\int_{0}^t \kf_{\Hm}(t,u)\D Z^\QQ_u+\int_{0}^t \lambda_u \kf_{\Hm}(t,u) \D u\right)\right\}
        \E^{-\int_{0}^t\gamma_u \D W^{\QQ\perp}_u -\half\int_{0}^t\gamma_u^2 \D u}\right|\Ff_s\right]\\
        & = \E^{\nu \int_{0}^t \lambda_u \kf_{\Hm}(t,u) \D u -\half\int_{0}^t\gamma_u^2 \D u}\EE^\QQ\left[\left.\exp\left\{\nu\int_{0}^t \kf_{\Hm}(t,s)\D Z^\QQ_s -\int_{0}^t\gamma_s \D W^{\QQ\perp}_s \right\}\right|\Ff_s\right],
    \end{align}
    where the second factor in the last term is just the conditional moment generating function of a Gaussian random variable. 
    Applying It\^o's isometry then, conditionally on~$\Ff_s$,
    the random variable 
    $\nu\int_{0}^t \kf_{\Hm}(t,s)\D Z^{\QQ}_s-\int_{0}^t\gamma_s \D W^{\QQ\perp}_s$is distributed as $\mathcal{N}\left(\mu,\sigma^2\right)$
    with
    \begin{align*}
        \mu & := \nu\int_{0}^s \kf_{\Hm}(t,u)\D Z^{\QQ}_u-\int_{0}^s\gamma_u \D W^{\QQ\perp}_u,\\
        \sigma^2 & := \nu^2\int_{s}^tk^2(t,u)\D u+\int_s^t \gamma_u^2 \D u -2\nu\rrho\int_s^t \kf_{\Hm}(t,u)\gamma_u \D u,
    \end{align*}
    since $Z^{\QQ} = \rho W^{\QQ}+\rrho W^{\QQ\perp}$. 
    Exploiting the identities above and reordering terms,
    \begin{align}
        \EE^\QQ\left[v_t\Ddgi_t|\Ff_s\right]  \
        & = \exp\Bigg\lbrace\nu\int_{0}^t \kf_{\Hm}(t,u)\lambda_u \D u +\overbrace{\nu\int_{0}^s \kf_{\Hm}(t,u)\D Z^{\QQ}_u-\int_{0}^s\gamma_u \D W^{\QQ\perp}_u}^{\mu}\\ \nonumber
        & \qquad  +\half\Bigg(\underbrace{\nu^2\int_{s}^tk^2(t,u)\D u+\int_s^t \gamma_u^2 \D u -2\nu\rrho\int_s^t \kf_{\Hm}(t,u)\gamma_u \D u}_{\sigma^2}-\int_{0}^t\gamma_u^2 \D u\Bigg)\Bigg\rbrace \\ \nonumber
        & = \exp\Bigg\lbrace -\nu\rrho\int_{s}^t \kf_{\Hm}(t,u)\gamma_u \D u \Bigg \rbrace 
        \overbrace{\exp\Bigg\lbrace -\int_{0}^s\gamma_u \D W^{\QQ\perp}_u  - \half\int_{0}^s\gamma_u^2 \D u\Bigg\rbrace}^{\EE^\QQ\left[\Ddgi_t|\Ff_s\right]}\\ \nonumber
        & \qquad \underbrace{+ \exp\Bigg\lbrace\nu\int_{0}^s \kf_{\Hm}(t,u)\D Z^{\QQ}_u +\frac{\nu^2}{2}\int_{s}^t \kf_{\Hm}^2(t,u)\D u + \nu\int_{0}^t \kf_{\Hm}(t,u)\lambda_u \D u  \Bigg\rbrace}_{\EE^\QQ[v_t|\Ff_s]},
    \end{align}
by using the decomposition of $\sigma^2$ as the sum of three terms, and so
\begin{align}\label{eq_dec_exps}
\EE^\QQ\left[v_t\Ddgi_t |\Ff_s\right] = \EE^\QQ\left[\Ddgi_t|\Ff_s\right]\EE^\QQ[v_t|\Ff_s]\exp\left\{-\nu\rrho\int_s^t \kf_{\Hm}(t,u)\gamma_u \D u\right\}.
\end{align}
Finally, going back to~\eqref{eq: Conditional Girsanov} and exploiting the identity in~\eqref{eq_dec_exps}, 
the result follows from
\begin{align}
\EE^\PP[v_t|\Ff_s]
 = \EE^\QQ[v_t|\Ff_s]\exp\left\{-\nu\rrho\int_s^t \kf_{\Hm}(t,u)\gamma_u \D u\right\}
 = \EE^\QQ[v_t|\Ff_s]\exp\left\{-\nu\int_s^t \kf_{\Hm}(t,u)\lambda_u \D u\right\}.
\end{align}
\end{proof}

%%%%%%%%%%%%%%%%%%%%%%%%%%%%%%%%%%%%%%%%%
\subsection{Estimating the risk premium in rough Bergomi}

In this section we work with the 
rough Bergomi model under~$\PP$ and its $\QQ$-version:
\begin{equation*}
\begin{array}{cl}
(\text{under }\PP) & \left\{
    \begin{array}{rl}
    \displaystyle \frac{\D S_t}{S_t} & =  \displaystyle \mu_t \D t + \sqrt{v_t}\, \D W^\PP_t,\\
    v_t & = \displaystyle \exp\left\{\nu Z^H_t\right\}, \\
    \end{array}
    \right.
\\
\\
(\text{under }\QQ) & \left\{
    \begin{array}{rl}
    \displaystyle \frac{\D S_t}{S_t} & =  \displaystyle r_t \D t + \sqrt{v_t}\, \D W^\QQ_t,\\
    v_t & = \displaystyle \xi_0(t)\exp\left\{\nu \left(\int_0^t \kf_{\Hm}(t-u) \D Z^\QQ_u+\int_0^t \kf_{\Hm}(t-u)\lambda_u \D u\right)\right\}.
    \end{array}
    \right.
\end{array}
\end{equation*}
Assuming~$\lambda$ deterministic, Theorem~\ref{thm: risk premium} gives an explicit procedure to compute the risk premium. 
In practice however, we are only able to observe variance swap quotes in discrete times, and hence it is natural to consider~$\lambda$ piecewise constant.

\begin{assumption}
Given a time partition $\{0=T_0<T_1,...,<T_n=T\}$, 
the deterministic process~$\lambda$ admits the piecewise constant representation 
\begin{equation}\label{eq:lambdaDis}        \lambda(t):=\sum_{i=1}^n \lambda_i\ind_{\{t\in[T_{i-1},T_i)\}},\quad \lambda_i\in\mathbb{R} \text{ for } i=1,\ldots,n.
    \end{equation}
Similarly the forward variance admits the  piecewise constant representation 
$\EE_{0}^{\QQ}[v_t] = \xi_0(t) := \sum_{i=1}^n \xi_i\ind_{\{t\in[T_{i-1},T_i)\}}$
with $\xi_i\in\RR$ 
for $i=1,\ldots,n$,
where 
    $\xi_i
        := \frac{\Vv_{T_i}T_i - \Vv_{T_{i-1}}T_{i-1}}{T_i-T_{i-1}}$
    and $\Vv_{T} := \EE^\QQ\left[\frac{1}{T}\int_0^{T} v_s \D s \right]$ is a market variance swap quote. 
\end{assumption}

We now estimate $\{\lambda_1,\cdots,\lambda_n\}$. 
The dataset consists of daily Eurostoxx variance swap quotes for maturities $\{\text{1M, 3M, 6M, 1Y, 2Y}\}$ 
(Figures~\ref{fig:VarSwap} and~\ref{fig:forwardVars}), 
while Figure~\ref{fig:realised_variance} shows the daily realised volatility obtained from Oxford-Man institute data.

\begin{figure}[H]
    \centering
    \includegraphics[scale=0.3]{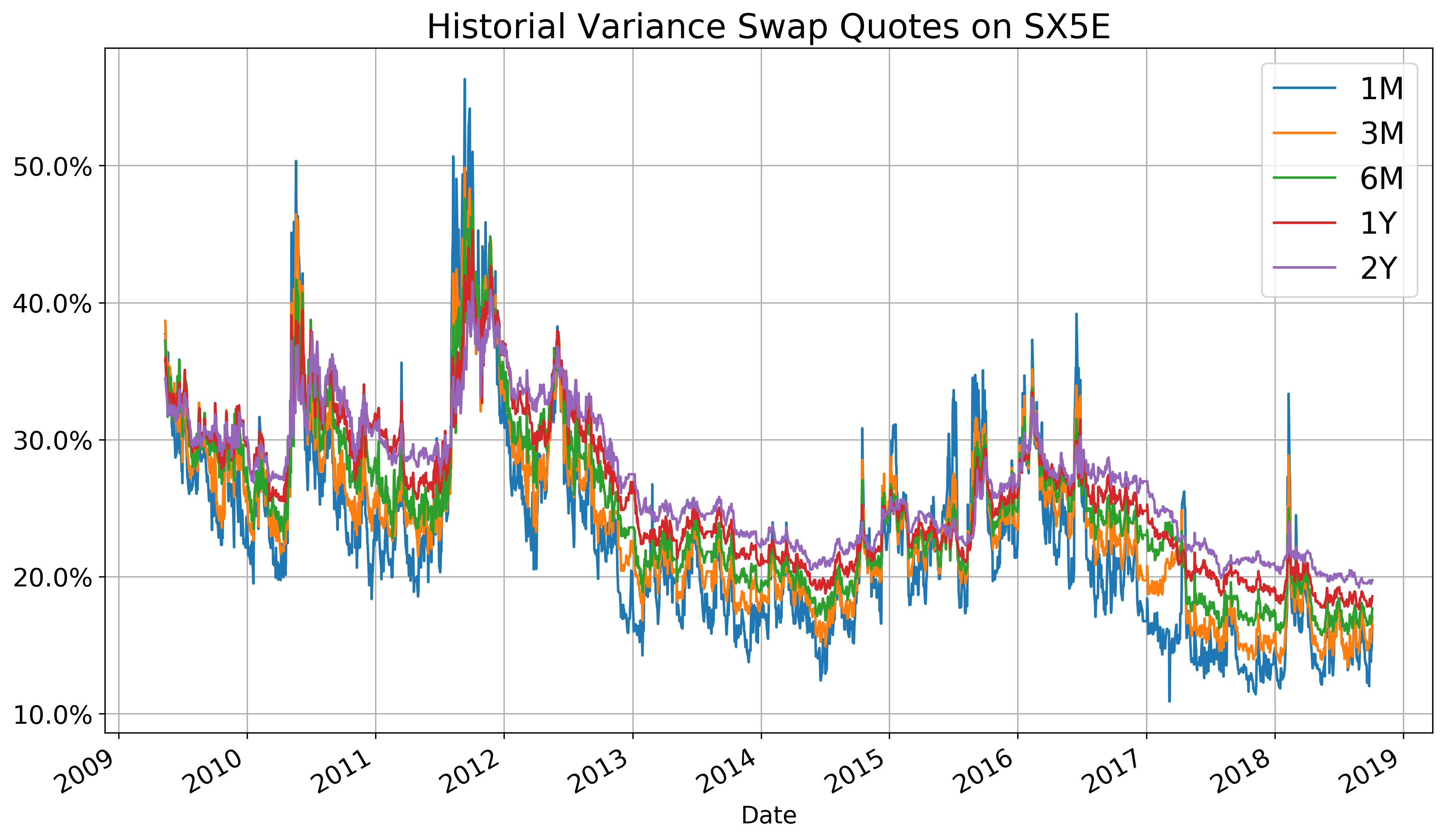}
    \caption{Variance Swap volatility daily quotes on SX5E}
    \label{fig:VarSwap}
\end{figure}

\begin{figure}[H]
    \centering
    \includegraphics[scale=0.3]{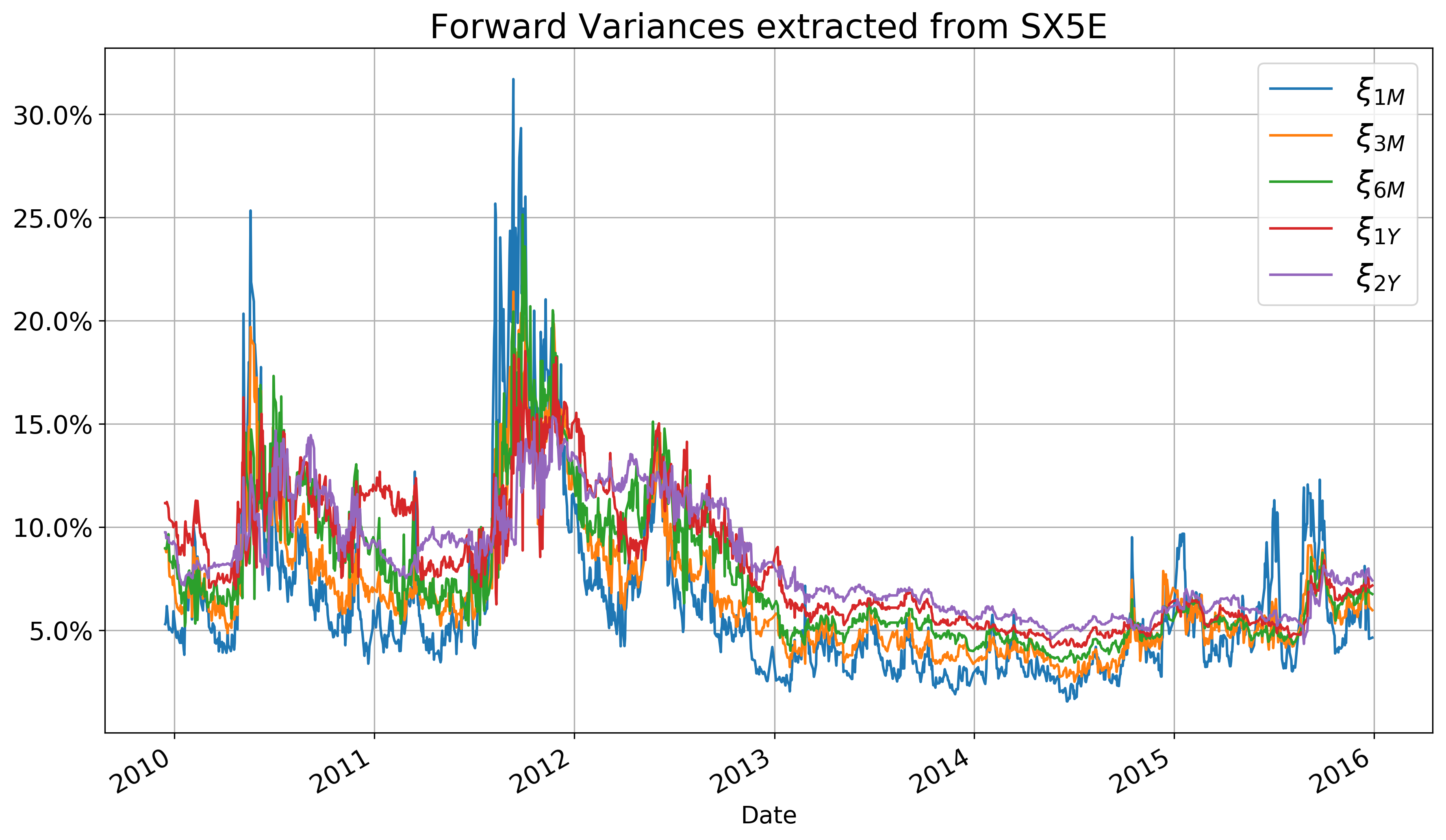}
    \caption{Forward variances extracted from variance swap quotes on SX5E}
    \label{fig:forwardVars}
\end{figure}

\begin{figure}[H]
    \centering
    \includegraphics[scale=0.3]{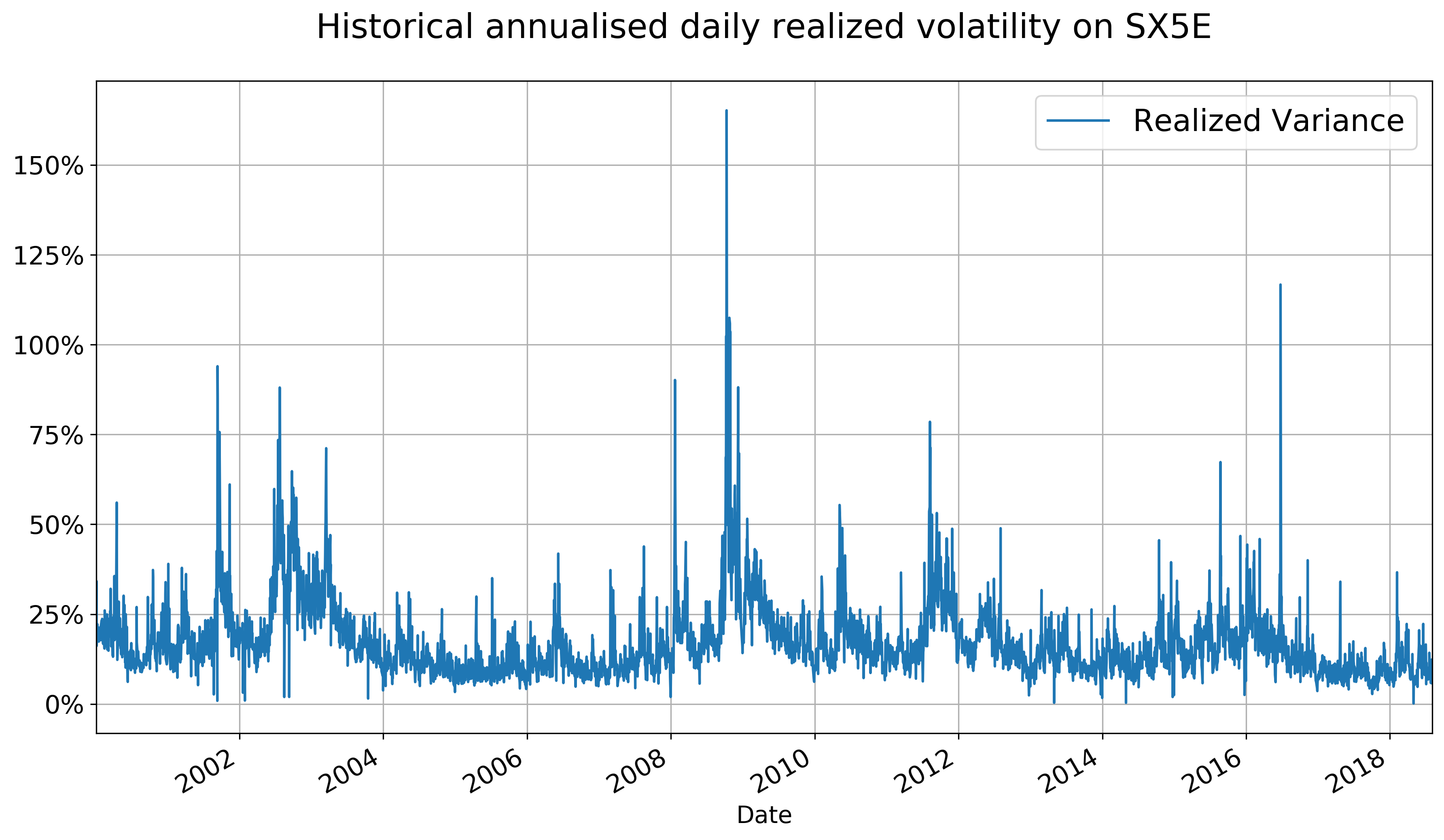}
    \caption{Annualised daily realised volatility on SX5E}\label{fig:realised_variance}
\end{figure}
In order to apply formula~\eqref{eq:riskPremium} we need the following ingredients:
\begin{equation}\label{eq:ingredients}
\text{Parameters }
H,\nu,\rho,
\qquad
 \EE^\PP[v_t|\mathcal{F}_0], \qquad
\EE^\QQ[v_t|\mathcal{F}_0].
\end{equation}
So far we have obtained $\EE^\QQ[v_t|\mathcal{F}_0]$ from Variance Swap market quotes. 
The next step is to estimate $(H,\nu,\rho)$ using historical time-series. 
Gatheral, Jaisson and Rosenbaum~\cite{gatheral2022volatility}, explain how to estimate~$\widehat{H}$ and~$\widehat{\nu}$ from daily volatility data (Figure~\ref{fig:realised_variance}), and we follow their approach using a 100-day rolling window (Figure~\ref{fig: Hnu})
and refer the reader to the original paper for details. 
\begin{figure}[H]
    \centering
    \includegraphics[scale=0.25]{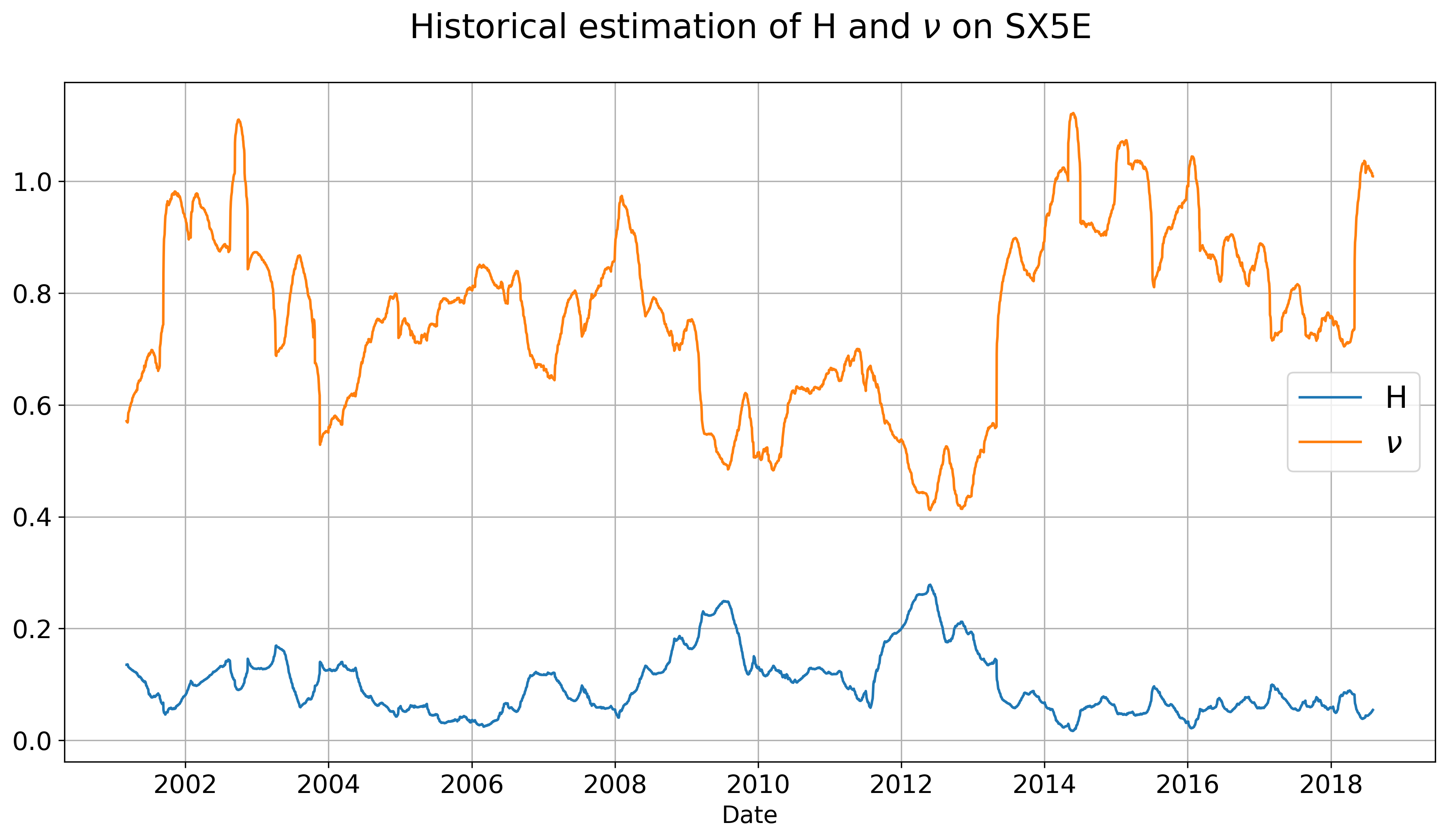}
    \caption{Estimated $\widehat{H}$ and $\widehat{\nu}$ on SX5E.}
    \label{fig: Hnu}
\end{figure} 

To estimate the correlation parameter we use
$\displaystyle \Corr\left(Z^H_t-Z^H_s,\int_s^t \D W_s\right)
 = \frac{\rho\sqrt{2H}}{H_+}$,
which allows us to estimate the correlation with the proxy
\begin{align}
    \widehat{\rho}
     = \frac{\widehat{H}+\half}{\sqrt{2\widehat{H}}}\ \widehat{\Corr}\left(\frac{\log(S_{t_i})-\log (S_{t_{i-1}})}{\sqrt{v_{t_{i-1}}}},\log(v_{t_i})-\log(v_{t_{i-1}})\right).
\end{align}
Figure~\ref{fig: corr} below displays the historical estimates using a estimation window of 100 days.
\begin{figure}[H]
    \centering
    \includegraphics[scale=0.25]{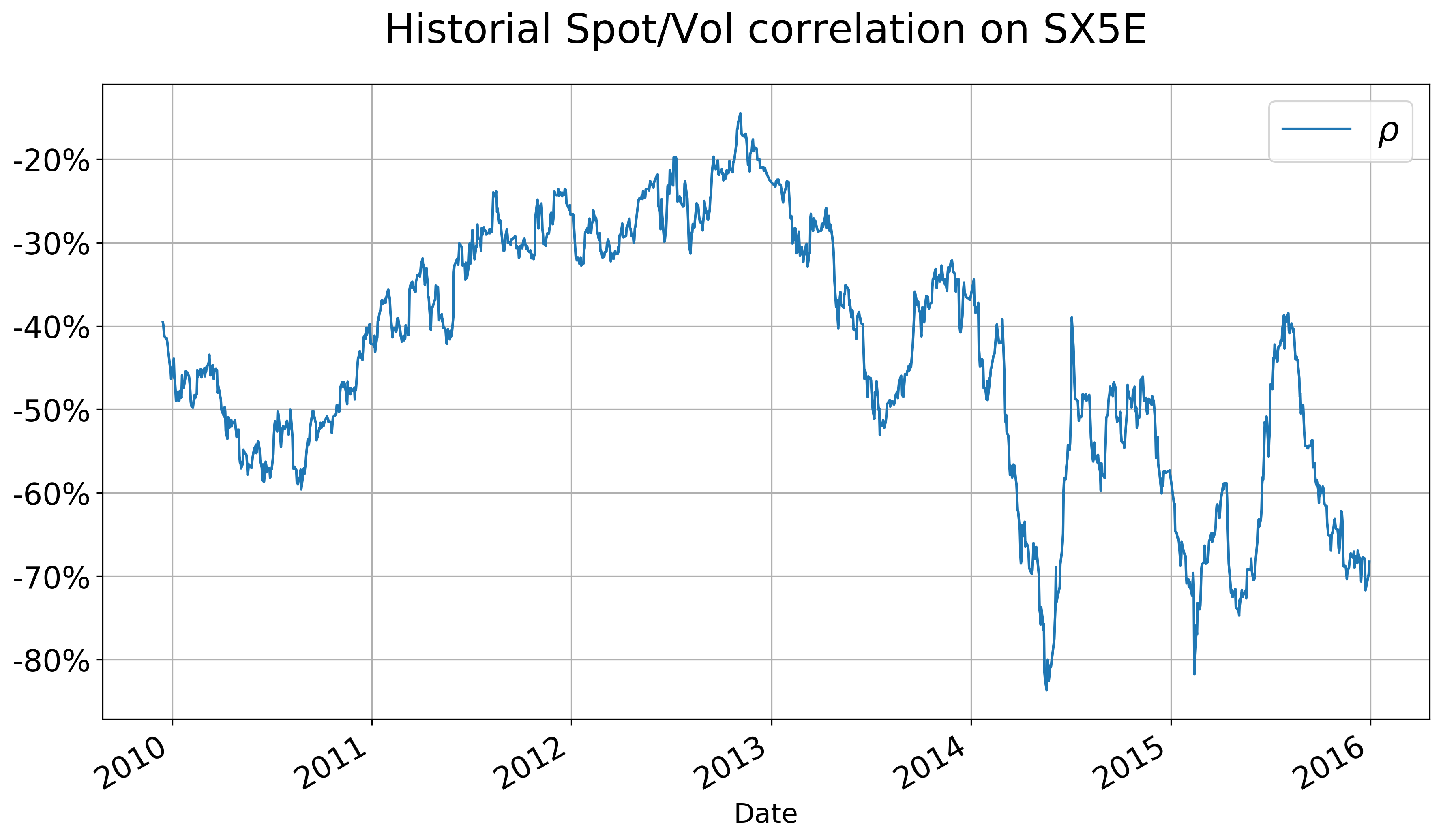}
    \caption{Daily correlation estimate on SX5E and realised volatility.}\label{fig: corr}
\end{figure}
To forecast volatility and obtain $\EE^\PP[v_t|\mathcal{F}_0]$, we proceed as in~\cite{gatheral2022volatility} and use the forecasting formula for the fractional Brownian motion due to Nuzman and Poor~\cite{nuzman2000linear}:

\begin{equation}\label{forecast under Q} 
    Z^H_{t+\Delta}|\Ff_t\sim\mathcal N\left(\frac{\cos(H\pi)}{\pi}\Delta^{H_+}\int_{0}^t\frac{Z^{H}_s \D s}{(t-s+\Delta)(t-s)^{H_+}},\frac{C_H\Delta^{2H}}{2H}\right).
\end{equation}

Finally, we orderly estimate $\lambda_i$ for each $i=1,\ldots,n$ using Theorem~\ref{thm: risk premium}
and the piecewise constant assumption~\eqref{eq:lambdaDis}, as
\begin{align}
    \sum_{j=1}^i \lambda_j\int_{T_{j-1}}^{T_j} \kf_{\Hm}(t,u) \D u=\frac{1}{\nu(1-\rho^2)}\log\left(\frac{\xi_0(T_i)}{\EE^\PP[v_{T_i}|\Ff_0]}\right).
\end{align}
Figure~\ref{fig:riskPremiumHistorial} shows the historical evolution of the risk premium process.

\begin{figure}[H]
    \centering
    \includegraphics[scale=0.3]{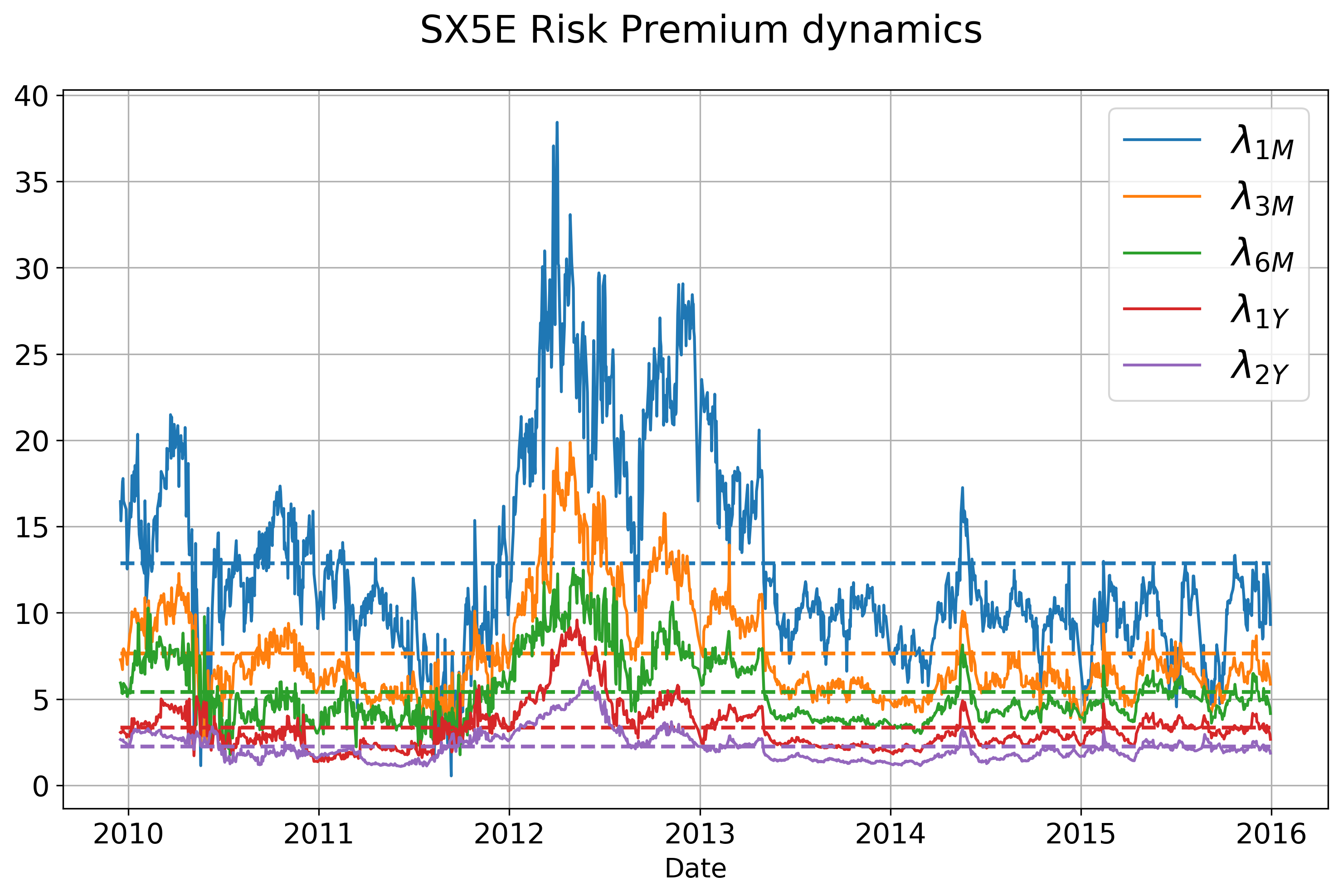}
    \caption{Daily SX5E risk premia; dashed lines represent means.}
    \label{fig:riskPremiumHistorial}
\end{figure}

\begin{remark}
We would like to emphasise that assessing the best method to estimate~\eqref{eq:ingredients} is beyond the scope of this paper.
However, as highlighted in the introduction, we stress the importance of Theorem~\ref{thm: risk premium} towards which this empirical work provides a first step. 
\end{remark}

\appendix
{
\section{Proof of Remark~\ref{rem_eq_c}}\label{app:rem_eq_c}
Following the ideas in~\cite[Proof of Theorem 1.1]{gassiat2019martingale}, 
we show that Assumption~\ref{assu:General}(iv) is guaranteed provided that Assumption~\ref{assump_0} and~\ref{assu:General}(i)-(ii) hold and the processes~$\mu, \gamma, r$ belong to~$\FF_{b}$. 
From the definition of the change of measure and the stopping time in~\eqref{eq: Pn Measure}, then $\Ddg_{t\wedge\tau_n} = \frac{\D\widehat{\PP}_n}{\D\PP}\bigg|_{\Ff_t}$.
%and recall the stopping time $\tau_n := \inf\{t\geq 0,\; Y_t=n\}$ from~\eqref{eq: Pn Measure}. 
For any $s\in\TTp$, the random function $f(x) := \frac{r_s-\mu_s}{\sqrt{\psi(s,x)}}$ is $\PP$-bounded on $(-\infty, a]$
for any $a>0$ since~$r$ and~$\mu$  are $\mathbb{P}$-bounded, and $\psi(s, \cdot)$ bounded away from zero on intervals of the form $(-\infty, a]$, by Assumption~\ref{assu:General}(i) together with the additional assumptions in Remark~\ref{rem_eq_c}.
Then, again by Proposition~\ref{prop:local martingale},
    \begin{equation}\label{eq:martingaleviaStop}
    1 = \EE\left[\Ddg_{T\wedge \tau_n}\right]
     = \EE\left[\Ddg_{T}\ind_{\{T< \tau_n\}}\right]
      + \EE\left[\Ddg_{\tau_n}\ind_{\{\tau_n\leq T\}}\right].
    \end{equation}
    The first term in~\eqref{eq:martingaleviaStop} converges  to $\EE\left[\Ddg_{T}\right]$ as~$n$ tends to infinity, yielding
    $$
    1-\EE\left[\Ddg_{T}\right]=\lim_{n\uparrow\infty}\EE\left[\Ddg_{\tau_n}\ind_{\{\tau_n\leq T\}}\right].
    $$
Girsanov's theorem implies
$\EE\left[\Ddg_{\tau_n}\ind_{\{\tau_n\leq T\}}\right] = \widehat{\PP}_n(\tau_n\leq T)$,
where $\widehat{\PP}_n$ is defined such that    
$\widehat{W}^n_t=W^\PP_t-\int_{0}^{t\wedge \tau_n} \ssf_u\D u$
is a $\widehat{\PP}_n$-Brownian motion. Then, under $\widehat{\PP}_n$, the process $Y$ becomes
$$
    Y_t
    =\widehat{Y}_t^n + \int_{0}^{t\wedge \tau_n} k(t,s) \Big(\rho \ssf_u +\rrho\;\gamma_u \Big)\D u
    =\widehat{Y}_t^n + \int_{0}^{t\wedge \tau_n} k(t,s) \lambda_u\D u,
$$
where $\displaystyle\widehat{Y}^n_t:=\int_{0}^t k(t,s)\D \widehat{Z}^n_s$ 
and
$$
    \widehat{Z}^n_t 
    = Z^\PP_t-\int_{0}^{t\wedge \tau_n} \Big(\rho \ssf_u + \rrho\;\gamma_u\Big)\D u
    = Z^\PP_t-\int_{0}^{t\wedge \tau_n} \lambda_u \D u,
$$
for $t\geq 0$, where $\widehat{Z}^n$ is a $\widehat{\PP}_n$-Brownian motion. 
Furthermore, by Assumption~\ref{assu:General}(ii), we have
\begin{align}\label{eq:ineqPQ}
\widehat{\PP}_n\left(\sup_{t\in \TT} Y_t\geq n\right)
    & = \widehat{\PP}_n\left(\sup_{t\in \TT}\left\{\widehat{Y}_t^n + \int_{0}^{t \wedge \tau_n} k(t,u)\lambda_u\D u\right\}\geq n\right)\\ \nonumber
    & \leq \widehat{\PP}_n\left(\sup_{t\in \TT} \widehat{Y}_t^n + \sup_{t\in \TT}\left\{\int_{0}^{t \wedge \tau_n} k(t,u)\lambda_u \D u\right\}  \geq n\right)\\\nonumber
    &\leq \widehat{\PP}_n\left(\sup_{t\in \TT} \widehat{Y}_t^n  \geq n - K_\TT \right),
\end{align} 
Inequality~\eqref{eq:ineqPQ}, in turn, implies $\widehat{\PP}_n(\tau_n\leq T) \leq \widehat{\PP}_n(\widehat{\tau}_n\leq T)$,
for $\widehat{\tau}_n:=\inf\{t\geq 0,\; \widehat{Y}_t^n=n-K_\TT\}$. 
Finally, since~$\widehat{Z}^n$ is a $\widehat{\PP}_n$-Brownian motion, we obtain
$$
    \lim_{n\uparrow\infty}\widehat{\PP}_n(\tau_n\leq T)\leq \lim_{n\uparrow\infty}\widehat{\PP}_n(\widehat{\tau}_n\leq T)
    = \lim_{n\uparrow\infty}\PP\left(\sup_{t\in \TT} Y_t\geq n-K_\TT\right)
    = 0,
$$
and it follows that $\displaystyle \frac{\D\QQ}{\D\PP}$ is indeed a true martingale and note that $\displaystyle\lim_{n\uparrow\infty}\widehat{\PP}_n=\QQ$. 
In the sense that the relation~\eqref{eq:QBM} holds between the~$\PP$ and~$\QQ$ Brownian motions.
}

%%%%%%%%%%%%%%%%%%%%%%%%%%%%%%%%%%%%%%%%%%%%%%%%%%%%
%%%%%%%%%%%%%%%%%%%%%%%%%%%%%%%%%%%%%%%%%%%%%%%%%%%%
\bibliography{References}
\bibliographystyle{siam}

\subsection*{Licence and Data Access statement}
\textit{For the purpose of open access, the author(s) has applied a Creative Commons Attribution (CC BY) licence (where permitted by UKRI, ‘Open Government Licence’ or ‘Creative Commons Attribution No-derivatives (CC BY-ND) licence’ may be stated instead) to any Author Accepted Manuscript version arising’}

\textit{The data underpinning this study were obtained from a combination of publicly and commercially available sources:
\begin{itemize}
\item Oxford-Man Institute of Quantitative Finance Realized Library: Available for academic research use at https://realized.oxford-man.ox.ac.uk/ upon registration and agreement to the terms of use.
\item Yahoo Finance: Market data were retrieved from https://finance.yahoo.com, a publicly accessible platform, under their data usage terms.
\item Bloomberg Terminal: Data sourced from Bloomberg are subject to licensing restrictions and cannot be shared publicly. Access to Bloomberg data requires an institutional subscription. 
\end{itemize}
Some components of the dataset are restricted due to commercial licensing agreements, and therefore cannot be made openly available. The decision to restrict access is to comply with the contractual obligations agreed with commercial data providers.
The present authors may be contacted for further information about the datasets used here.
}
\end{document}